\documentclass[journal,twoside,web]{ieeecolor}
\usepackage{tabularx} 
\usepackage{diagbox}
\usepackage{graphicx}
\usepackage{mathrsfs}
\usepackage{generic}
\usepackage{cite}
\usepackage{amssymb,amsfonts}
\usepackage{algorithmic}
\usepackage{textcomp}
\def\BibTeX{{\rm B\kern-.05em{\sc i\kern-.025em b}\kern-.08em
    T\kern-.1667em\lower.7ex\hbox{E}\kern-.125emX}}
\usepackage[dvipsnames]{xcolor}
\usepackage{color, soul}
\sethlcolor{Bittersweet}

\newcommand{\comb}{\genfrac{(}{)}{0pt}{}}

\usepackage{textcomp}
\def\BibTeX{{\rm B\kern-.05em{\sc i\kern-.025em b}\kern-.08em
    T\kern-.1667em\lower.7ex\hbox{E}\kern-.125emX}}
\usepackage{float}
\usepackage[tbtags]{amsmath}
\usepackage[tbtags]{mathtools}
\usepackage{accents}
\usepackage{multicol}
\usepackage{multirow}
\usepackage[ruled,vlined]{algorithm2e}
\usepackage{cite}
\usepackage{pgfplots}

\usepackage{array,multirow}
\usepackage{hhline}
\usepackage{arydshln}

\usepackage[hidelinks=true]{hyperref}

\newtheorem{theorem}{Theorem}
\newtheorem{corollary}{Corollary}
\newtheorem{lemma}{Lemma}
\newtheorem{example}{Example}
\newtheorem{assumption}{Assumption}

\newtheorem{problem}{Problem}
\newtheorem{alg}{Algorithm}
\newtheorem{remark}{Remark}

\newtheorem{definition}{Definition}

\makeatletter
\def\fnum@figure{\textcolor{subsectioncolor}{\sf Fig.~\thefigure}}
\def\fnum@table{\textcolor{subsectioncolor}{\sf TABLE~\thetable}}
\makeatother

\begin{document}
\title{A Framework for Prescribed-Time Control Design via Time-Scale Transformation}
\author{Amir Shakouri, \IEEEmembership{Member, IEEE}, and Nima Assadian, \IEEEmembership{Senior Member, IEEE}
\thanks{The authors are with the Department of Aerospace Engineering, Sharif University of Technology, Tehran, Iran (e-mail: \href{mailto:a_shakouri@outlook.com}{a$\_$shakouri@outlook.com}; \href{mailto:assadian@sharif.edu}{assadian@sharif.edu}).}}

\maketitle

\pagestyle{empty}
\thispagestyle{empty}

\begin{abstract}
This letter presents a unified framework for the design of prescribed-time controllers under time-varying input and state constraints for normal-form unknown nonlinear systems with uncertain input gain. The proposed approach is based on a time-domain mapping method by which any infinite-time system can be corresponded to a prescribed-time system and vice versa. It is shown that the design of a constrained nonasymptotic prescribed-time controller can be reduced to the asymptotic control design for an associated constrained infinite-time system. Fa\`a di Bruno's formula and Bell polynomials are used for a constructive representation of the associated infinite-time system. The presented results are not confined to a particular mapping function, which adds to the flexibility of the proposed scheme. It is shown that necessary and sufficient conditions on the uniform (practical) prescribed-time stability and attractivity can be obtained as corollaries of the main result.
\end{abstract}
\begin{IEEEkeywords}
Prescribed-time control, state and input constraints, time-transformation mapping, Fa\`a di Bruno's formula, Bell polynomials.
\end{IEEEkeywords}

\section{Introduction}
\label{sec:I}

\IEEEPARstart{T}{he} nonasymptotic control of linear and nonlinear systems has been a challenging field of research over the past two decades, and different approaches have been developed to tackle this problem. The user command over the time of convergence divides the nonasymptotic schemes into three major categories of finite-time, fixed-time, and prescribed-time methods. In the finite-time approach, the system state converges to the origin at an unknown finite time depending on the initial conditions. Fixed-time controllers maintain the mentioned features of the finite-time methods, while in this category, an upper bound can be estimated for the convergence time, independently of the initial conditions. The prescribed-time control is a time-varying scheme where the system is forced to converge precisely at the commanded moment. The robustness of the prescribed-time controllers (PTCs) against matched disturbances and their adjustable settling time has many applications in practical engineering systems, especially when multiple time-varying systems cooperate in a dynamically changing environment. The safety and accuracy of uncertain industrial processes, aerospace systems, robotic systems, etc., can be enhanced by imposing hard time constraints on the control tasks.

The nonsmooth feedback control \cite{in2,in3} and the terminal sliding mode control \cite{in5,in6,in7} are the most popular methods for finite-time control of nonlinear systems. To deal with uncertainties, different approaches are used in the literature for achieving an adaptive finite-time control and tracking scheme \cite{in8,in9,in11}. The fixed-time control is first proposed for linear systems by Polyakov \cite{in12}. Since then, many investigations have been carried out to solve the fixed-time control problem of nonlinear systems. The non-singular terminal sliding mode control is proposed in \cite{in13} for a class of second-order nonlinear systems with matched disturbances. An output feedback scheme for disturbed double-integrator systems is addressed in \cite{in14}, while the stabilization of high-order integrator systems with mismatched disturbances is studied in \cite{in15}. 

Prescribed-time stabilization of uncertain nonlinear systems is originally proposed by Song et al. in \cite{in18}. In the prescribed-time control methods, a time transformation mapping from the infinite time scale onto a desired finite time scale has been the central idea to achieve a time-varying controller. The controllers obtained through this approach have also demonstrated disturbance suppressing behavior against non-vanishing and parametric uncertainties \cite{in20,IrscheidBleymehlRudolph}. A similar approach has been used in the literature, called the generalized time transformation method \cite{tran2020finite}, which essentially incorporates the same idea. PTCs for systems with matched uncertainty and uncertain input gain have been studied in \cite{krishnamurthy2020dynamic,krishnamurthy2020prescribed}. Prescribed-time stabilization with some robustness to unmatched uncertainties and measurement noise is investigated in \cite{chitour2020stabilization,ye2021prescribed}. A class of PTCs with linear decay rate is proposed in \cite{shakouri2021prescribed} for unknown nonlinear systems with uncertain input gain. In addition, input-constrained PTCs are analyzed in \cite{garg2020prescribed}. PTC design using parametric Lyapunov equations is studied in \cite{zhou2020finite,zhou2021prescribed}. The prescribed-time extremum seeking control is proposed by \cite{poveda2021non}, and the inverse optimality of PTCs is studied in \cite{li2021stochastic}. Moreover, a predictor-feedback PTC for systems with input delay is proposed by \cite{espitia2021predictor}. In terms of applications, the reader can refer to the PTCs developed for robotic systems in \cite{cui2021prescribed,shakouri2021prescribedeul,cao2021practical} and those tested empirically in \cite{wang2019adaptive,zhang2021practical}. 

This letter presents a unified framework for the prescribed-time control of unknown nonlinear systems in the form of a chain of integrators with matched unmodeled dynamics subject to time-varying state and input constraints. The proposed approach is general so that the majority of the previous methods can be considered particular types of the presented scheme. In fact, this letter extends the design approach studied by the authors in \cite{shakouri2021prescribedeul} to higher-order nonlinear systems. In this study, using Fa\`a di Bruno's formula and Bell polynomials, we introduce a constructive method to represent the associated infinite-time version of a prescribed-time system obtained by a user-defined sufficiently differentiable time transformation mapping function. Then, we show that the nonasymptotic prescribed-time control design can be reduced to the conventional asymptotic infinite-time control design for the associated infinite-time system. Also, we show how the constraints' role is preserved under the mapping process.

\section{Preliminaries}
\label{sec:II}

\subsection{Notations}
\label{sec:II-1}

Let $\mathbb{R}^{m,n}$ and $\mathbb{R}^n$ denote the space of $m \times n$ real matrices and $n$-dimensional vectors, respectively. The $n$-dimensional identity matrix is denoted by $\mathbb{I}_n$. The $i$th entry of vector $r$ and the $ij$th entry of matrix $M$ are referred to by $r_i$ and $M_{ij}$, respectively. The inverse and the transpose of matrix $M$ are denoted by $M^{-1}$ (if the inverse exists) and $M^T$, respectively. The inverse function of $f(\cdot)$ is denoted by $f^{-1}(\cdot)$ (if the inverse exists). The $i$th derivative of function $f$ with respect to its argument is shown by $f^{(i)}$. The space of continuous functions that have continuous first $r$ derivatives on their domain is denoted by $C^r$ and the space of functions that are smooth on their domain are shown by $C^\infty$. The 2-norm (for vectors and matrices) and the absolute value (for scalars) are denoted by $\|\cdot\|$ and $|\cdot|$, respectively. 

The binomial coefficient, indexed by a pair of integers $n$ and $m$, is denoted by $\comb{n}{m}$. The number of partitions of set $\{1,2,\ldots,n\}$ with block sizes $c_1,\ldots,c_m$ is referred to by $\comb{n}{c_1,\ldots,c_m}$. The partial Bell polynomials are denoted by $B_{n,m}(s_1,s_2,\ldots,s_{n-m+1})$ that can be recurrently obtained as \cite{comtet2012advanced}:
\begin{equation}
\label{eq:note-1}
B_{n,m}=\sum_1^{n-m+1}\comb{n-1}{m-1}s_iB_{n-i,m-1},
\end{equation}
such that $B_{0,0}=1$, $B_{n,0}=0$ for $n\geq1$, and $B_{0,m}=0$ for $m\geq1$\footnote{Subscript notations $X_{ij}$ and $Y_{i,j}$ should not be confused throughout the letter as the former addresses an entry in matrix $X$, while the latter denotes the dependency of $Y$ on two indices $i$ and $j$.}.

\subsection{System description and problem statement}

Consider a nonlinear system in its normal form as follows: 
\begin{equation}
\label{eq:1}
\left\{
\begin{array}{lcl}
\dot{x}_i&=&x_{i+1};\quad i=1,2,\ldots,n-1 \\
\dot{x}_n&=&f(x,u,t)+g(x,t)u
\end{array}\right.,
\end{equation}
where $x=[x_1,\ldots,x_n]^T\in\mathbb{R}^{n}$ and $u\in\mathbb{R}$ are the state vector and the control input, respectively. Functions $f(\cdot,\cdot,\cdot):\mathbb{R}^{n}\times\mathbb{R}\times[0,\infty)\rightarrow\mathbb{R}$ and $g(\cdot,\cdot):\mathbb{R}^{n}\times[0,\infty)\rightarrow\mathbb{R}$ stand for the unknown dynamics and the known nonzero control input gain, respectively. Note that any uncertainty in the input gain can be accounted in function $f(x,u,t)$.

\begin{assumption}
\label{ass:1}
Suppose that system \eqref{eq:1} is controllable, its dimension is known, and there exists a known set-valued function $\mathcal{F}(x,u,t)\subset\mathbb{R}$ such that $f(x,u,t)\in\mathcal{F}(x,u,t)$ for all $x\in\mathbb{R}^n$, $u\in\mathbb{R}$, and $t\in[0,\infty)$.
\end{assumption}

The main problem of the letter can be stated as follows:

\begin{problem}
\label{prob:1}
Let $t_0$ denote the initial time and $h(\cdot,\cdot,\cdot):\mathbb{R}^{n}\times\mathbb{R}\times[0,\infty)\rightarrow\mathbb{R}^m$ is defined by the user. Suppose that Assumption \ref{ass:1} holds. Find a function $\pi(\cdot,\cdot,\cdot):\mathbb{R}^{n}\times[0,\infty)\times(0,\infty)\rightarrow\mathbb{R}$ and a set $\mathcal{T}\subseteq(0,\infty)$ such that for every $\tau\in\mathcal{T}$ the closed-loop solution constructed by $u=\pi(x,t,\tau)$ and system \eqref{eq:1} satisfies a feasible control objective as:
\begin{equation}
\label{eq:prob:1-1}
h(x,u,t)\in\mathcal{H}(t),\hspace{2mm}\forall t\in[t_0,t_0+\tau),
\end{equation}
where $\mathcal{H}\subseteq\mathbb{R}^m$ is an arbitrary set-valued function of time.
\end{problem}

Problem \ref{prob:1} seeks a PTC for an unknown nonlinear system where the control objective is considered in its most general form, allowing a wide range of constraints involving state, input, and time. For instance, time-varying input and output constraints such as $|u|<\delta_{\mathrm{in}}(t)$ and $|x_1|<\delta_{\mathrm{out}}(t)$ can be represented by $h=u$ and $h=x_1$ with $\mathcal{H}(t)=\{h:|h|<\delta_{\mathrm{in}}(t)\}$ and $\mathcal{H}(t)=\{h:|h|<\delta_{\mathrm{out}}(t)\}$, respectively. Moreover, by means of condition \eqref{eq:prob:1-1}, user-defined decay rates (e.g., faster than exponential) or maximum energy efforts (e.g., considering a functional $h=\int_{t_0}^
{t}(x^T(\nu)Qx(\nu)+u^T(\nu)Ru(\nu))\mathrm{d}\nu$) can be imposed to the system.

\subsection{Basic definitions and formulations}

The definitions and formulas introduced in this subsection are frequently used to present and prove the main results of this letter. 

\begin{definition}[see \cite{shakouri2021prescribedeul}]
\label{def:1}
A $C^{n}$ function $\kappa(\cdot):[0,\tau)\rightarrow[0,\infty)$ is said to be class $\mathcal{K}$ (or $\kappa\in\mathcal{K}(\tau)$) if $\dot{\kappa}(t)>0$, $\kappa(0)=0$, and $\lim_{t\rightarrow\tau^-}\kappa(t)=\infty$. This class is a special surjective form of a more general case used in the literature under the same name (see \cite[Definition 4.2]{khalil2002nonlinear}).
\end{definition}

\begin{definition}[see \cite{shakouri2021prescribedeul}]
\label{def:2}
A $C^{n}$ function $\mu(\cdot):[0,\infty)\rightarrow[0,\tau)$ is said to be class $\mathcal{M}$ (or $\mu\in\mathcal{M}(\tau)$) if its inverse function is class $\mathcal{K}$ (i.e., $\mu^{-1}\in\mathcal{K}(\tau)$). Therefore, $\mu$ is a continuous function subject to $\dot{\mu}(t)>0$, $\mu(0)=0$, and $\lim_{t\rightarrow\infty}\mu(t)=\tau$.
\end{definition}

\begin{example}
\label{ex:1}
The following function is of class $\mathcal{K}$ for any $a_i>0$, $b_i>1$, and $n\geq1$:
\begin{equation}
\label{eq:ex:1-1}
\kappa(t)=-\sum_{i=1}^na_i\log_{b_i}\left(1-\frac{t}{\tau}\right),
\end{equation}
and the following function is a class $\mathcal{M}$ function for any $a_i>1$, $b_i>0$, and $n\geq1$:
\begin{equation}
\label{eq:ex:1-2}
\mu(t)=\frac{1}{n}\sum_{i=1}^n\tau(1-a_i^{-b_it}).
\end{equation}
\end{example}

We use $s(\cdot):\mathbb{R}\rightarrow\mathbb{R}$ as a typical sufficiently differentiable function in our definitions, which can be replaced by any other functions such as $\kappa(t)$ and $\mu(t)$. Define the $n\times n$-dimensional matrix-valued functional $\mathfrak{B}_n[s]$ acting on a $C^n$ function $s(t)$ by the following entrywise rule:
\begin{equation}
\label{eq:III-3}
\mathfrak{B}_n[s]_{ij}=\left\{\begin{array}{lcl}
B_{i-1,j-1}(r_1[s],\ldots,r_{i-j+1}[s]) & \hspace{2mm} & j\leq i\leq n \\
0 & \hspace{2mm} & \mathrm{otherwise} 
\end{array}\right.,
\end{equation}
with the functional
\begin{equation}
\label{eq:III-4}
r_{k+1}[s]=\frac{1}{\dot{s}^{k+1}(t)}\sum_{m=0}^k\frac{(-1)^{m}}{\dot{s}^{m}}R_{k,m}(s^{(2)},\ldots,s^{(k-m+2)}),
\end{equation}
where $R_{k,m}(s^{(2)},\ldots,s^{(k-m+2)})$ can be evaluated by the sum
\begin{equation}
\label{eq:III-5}
R_{k,m}=\frac{1}{m!}\sum\comb{k+m}{c_1,\ldots, c_m}\prod_{l=1}^{m}s^{(c_l)},\ \mathrm{s.t.}\ R_{0,0}=1,
\end{equation}
in which the sum is taken over all integers $c_l\geq2$, $l=1,\ldots,m$ satisfying $\sum_{l=1}^ic_l=k+m$.

In addition, define the $(n-1)$-dimensional vector-valued functional $\mathfrak{b}_{n-1}[s]$ as follows:
\begin{equation}
\label{eq:III-5}
\mathfrak{b}_{n-1}[s]=[\mathfrak{B}_{n}[s]_{n,1},\ \ldots,\ \mathfrak{B}_{n}[s]_{n(n-1)}]^T.
\end{equation}

\begin{example}
\label{ex:2}
It can be verified that $r_1[s(t)]=1/\dot{s}(t)$ and $r_2[s(t)]=-s^{(2)}(t)/\dot{s}^{3}(t)$. Thus, matrix $\mathfrak{B}_3$ and vector $\mathfrak{b}_2$ can be evaluated as:
\begin{equation}
\label{eq:ex:2-1}
\mathfrak{B}_3[s(t)]=\left[\begin{array}{ccc}
1 & 0 & 0 \\
0 & \frac{1}{\dot{s}(t)} & 0 \\
0 & -\frac{{s}^{(2)}(t)}{\dot{s}^3(t)} & \frac{1}{\dot{s}^2(t)}
\end{array}\right],
\end{equation}
\begin{equation}
\label{eq:ex:2-2}
\begin{split}
\mathfrak{b}_2[s(t)]&=[0,\ -s^{(2)}(t)/\dot{s}^3(t)]^T.
\end{split}
\end{equation}
Tables \ref{table:A1} and \ref{table:A2} present some evaluations of functionals $R_{i,j}$ and $r_j$, respectively, enough for obtaining $\mathfrak{B}_{n+1}$ and $\mathfrak{b}_n$ for $n=1,\ldots,5$, useful for systems up to fifth-order.
\end{example}
\begin{table}[h!]
\centering
\caption{Evaluation of $R_{j,i}(s^{(2)},\ldots,s^{(j-i+2)})$.}
\label{table:A1}
\resizebox{\columnwidth}{!}{\begin{tabular}{|c|ccccc|}
\hline
\backslashbox{$j$}{$i$} & $0$ & $1$ & $2$ & $3$ & $4$ \\
\hline
$0$ & $1$ & \multicolumn{1}{c}{} & \multicolumn{1}{c}{} & \multicolumn{1}{c}{} & \\
$1$ & $0$ & $s^{(2)}$ & \multicolumn{1}{c}{} & \multicolumn{1}{c}{} & \\
$2$ & $0$ & $s^{(3)}$ & $(s^{(2)})^2$ & \multicolumn{1}{c}{} & \\
$3$ & $0$ & $s^{(4)}$ & $10s^{(2)}s^{(3)}$ & $15(s^{(2)})^3$ & \\
$4$ & $0$ & $s^{(5)}$ & $15s^{(2)}s^{(4)}+10(s^{(3)})^2$ & $105(s^{(2)})^2s^{(3)}$ & $105(s^{(2)})^4$ \\
\hline
\end{tabular}}
\end{table}
\begin{table}[h!]
\centering
\caption{Evaluation of $r_{j}[s(t)]$.}
\label{table:A2}
\resizebox{\columnwidth}{!}{\begin{tabular}{|c|l|}
\hline
Functional & Value \\
\hline
$r_1[s]$ & $\frac{1}{\dot{s}}$ \\
$r_2[s]$ & $-\frac{s^{(2)}}{\dot{s}^{3}}$ \\
$r_3[s]$ & $-\frac{s^{(3)}}{\dot{s}^4}+\frac{(s^{(2)})^2}{\dot{s}^5}$ \\
$r_4[s]$ & $-\frac{s^{(4)}}{\dot{s}^5}
+\frac{10s^{(2)}s^{(3)}}{\dot{s}^6}
-\frac{15(s^{(2)})^3}{\dot{s}^7}$ \\
$r_5[s]$ & $-\frac{s^{(5)}}{\dot{s}^6}
+\frac{15s^{(2)}s^{(4)}+10(s^{(3)})^2}{\dot{s}^7}
-\frac{105(s^{(2)})^2s^{(3)}}{\dot{s}^8}
+\frac{105(s^{(2)})^4}{\dot{s}^9}$ \\
\hline
\end{tabular}}
\end{table}

\section{Main Results}
\label{sec:III}

Consider the following theorem as the central result of this letter:

\begin{theorem}
\label{th:1}
Suppose that Assumption \ref{ass:1} holds. Let $t_0$ denote the initial time and $\Delta t=t-t_0$. Also, let $\mu\in\mathcal{M}(\tau)$, $\kappa=\mu^{-1}\in\mathcal{K}(\tau)$, and $\pi_0(\cdot,\cdot):\mathbb{R}^n\times[0,\infty)\rightarrow\mathbb{R}$ be sufficiently differentiable user-defined functions by which $\pi(\cdot,\cdot,\cdot):\mathbb{R}^n\times[0,\infty)\times[0,\infty)\rightarrow\mathbb{R}$ is defined as follows:
\begin{equation}
\label{eq:th:1-1}
\pi(x,t,\tau)=\frac{\dot{\kappa}^n(\Delta t)}{g(x,t)}\left(\pi_0(x_{\kappa},t_{\kappa})-\mathfrak{b}_{n}^T[\kappa(\Delta t)]x\right),
\end{equation}
where
\begin{equation}
\label{eq:th:1-2}
x_{\kappa}=\mathfrak{B}_n[\kappa(\Delta t)]x,
\end{equation}
\begin{equation}
\label{eq:th:1-3}
t_{\kappa}=\kappa(\Delta t)+t_0.
\end{equation}
Then, the solution of system \eqref{eq:1} under $u=\pi(x,t,\tau)$ for every $f(x,u,t)\in\mathcal{F}(x,u,t)$ and $\tau\in\mathcal{T}$ with initial condition $x(t_0)$ satisfies $h(x,u,t)\in\mathcal{H}(t)$ for all $t\in[t_0,t_0+\tau)$ if and only if the solution of the following system:
\begin{equation}
\label{eq:th:1-4}
\left\{
\begin{array}{lcl}
\dot{\xi}_i&=&\xi_{i+1};\hspace{4mm} i=1,2,\ldots,n-1 \\
\dot{\xi}_n&=&\dot{\mu}^n(\Delta t)f_{\mu}(\xi,t)+\pi_0(\xi,t)
\end{array}\right.
\end{equation}
for every $f_{\mu}(\xi,t)\in\mathcal{F}(\xi_{\mu},u_{\mu},t_{\mu})$ and $\tau\in\mathcal{T}$ with initial condition $\xi(t_0)=\mathfrak{B}_n[\kappa(\Delta t)]_{t=t_0}x(t_0)$ satisfies: 
\begin{equation}
\label{eq:th:1-5}
h(\xi_{\mu},u_{\mu},t_{\mu})\in\mathcal{H}(t_{\mu})
\end{equation}
for all $t\in[t_0,\infty)$, where 
\begin{equation}
\label{eq:th:1-6}
\xi_{\mu}=\mathfrak{B}_n[\mu(\Delta t)]\xi,
\end{equation}
\begin{equation}
\label{eq:th:1-7}
u_{\mu}=\frac{1}{g(\xi_\mu,t_\mu)}\left(\frac{\pi_0(\xi,t)}{\dot{\mu}^n(\Delta t)}+\mathfrak{b}_{n}^T[\mu(\Delta t)]\xi\right),
\end{equation}
\begin{equation}
\label{eq:th:1-8}
t_{\mu}=\mu(\Delta t)+t_0.
\end{equation}
\end{theorem}

The following corollary elaborates on the application of Theorem \ref{th:1} to PTC design with a state-constrained performance objective.

\begin{corollary}
\label{corollary:2}
Suppose that Assumption \ref{ass:1} holds. Let $t_0$ denote the initial time, $\Delta t=t-t_0$, $\mu\in\mathcal{M}(\tau)$, and $\kappa=\mu^{-1}\in\mathcal{K}(\tau)$. Also, let $\zeta(\cdot):[0,\tau)\rightarrow[0,\infty)$ be a user-defined function of time. Then, for every $f(x,u,t)\in\mathcal{F}(x,u,t)$ and $\tau\in\mathcal{T}$ the solution of system \eqref{eq:1} under $u=\pi(x,t,\tau)$ satisfies $\|x(t)\|\leq\sigma\|x(t_0)\|\zeta(\Delta t)$ for some $\sigma\geq1$ at all $t\in[t_0,t_0+\tau)$ if and only if for every $f_{\mu}(\xi,t)\in\mathcal{F}(\xi_{\mu},u_{\mu},t_{\mu})$ and $\tau\in\mathcal{T}$ the solution of system \eqref{eq:th:1-4} satisfies:
\begin{equation}
\label{eq:cor:2-1}
\left\|\mathfrak{B}_n[\mu(\Delta t)]\xi(t)\right\|\leq\sigma\|\mathfrak{B}_n[\mu(\Delta t)]_{t=t_0}\xi(t_0)\|\zeta(\mu(\Delta t))
\end{equation}
for some $\sigma\geq1$ at all $t\in[t_0,\infty)$.
\end{corollary}
\begin{proof}
Given Theorem \ref{th:1}, one may consider $h(x)=\|x\|$ and $\mathcal{H}(t)=\{h:0\leq h\leq\sigma\|x(t_0)\|\zeta(\Delta t)\}$ which means that $\mathcal{H}(t_{\mu})=\mathcal{H}(\mu(\Delta t)+t_0)=\{h:0\leq h\leq\sigma\|x(t_0)\|\zeta(\mu(\Delta t)\}$. Therefore, according to \eqref{eq:th:1-5}, one may find $\|\xi_{\mu}\|\in\mathcal{H}(t_{\mu})$ as a necessary and sufficient condition, that is inequality \eqref{eq:cor:2-1}. 
\end{proof}

For PTC design subject to a time-varying constraint on the control input, Theorem \ref{th:1} reduces to the following corollary.

\begin{corollary}
\label{corollary:3}
Suppose that Assumption \ref{ass:1} holds. Let $t_0$ denote the initial time, $\Delta t=t-t_0$, $\mu\in\mathcal{M}(\tau)$, and $\kappa=\mu^{-1}\in\mathcal{K}(\tau)$. Also, let $\upsilon(\cdot):[0,\tau)\rightarrow[0,\infty)$ be a user-defined function of time. Then, for every $f(x,u,t)\in\mathcal{F}(x,u,t)$ and $\tau\in\mathcal{T}$ the solution of system \eqref{eq:1} under $u=\pi(x,t,\tau)$ satisfies $|u|\leq\upsilon(\Delta t)$ at all $t\in[t_0,t_0+\tau)$ if and only if for every $f_{\mu}(\xi,t)\in\mathcal{F}(\xi_{\mu},u_{\mu},t_{\mu})$ and $\tau\in\mathcal{T}$ the solution of system \eqref{eq:th:1-4} satisfies the following condition at all $t\in[t_0,\infty)$:
\begin{equation}
\label{eq:cor:3-1}
\left|\frac{\pi_0(\xi,t)}{\dot{\mu}^n(\Delta t)}+\mathfrak{b}_{n}^T[\mu(\Delta t)]\xi\right|\leq\upsilon(\mu(\Delta t)).
\end{equation}
\end{corollary}
\begin{proof}
Consider $h(u)=|u|$ and define $\mathcal{H}(t)=\{h:0\leq h\leq\upsilon(\Delta t)\}$. The rest is similar to the proof of Corollary \ref{corollary:2}.
\end{proof}

In the following, we discuss the application of Theorem \ref{th:1} to prescribed-time stability/attractivity analysis.

\begin{definition}
\label{def:PTS}
System \eqref{eq:1} under controller $u=\pi(x,t,\tau)$ is called
\begin{enumerate}
\item \textit{globally practically prescribed-time attractive at every $\tau\in\mathcal{T}$ with error $\varsigma>0$}, if for every $\tau\in\mathcal{T}$ and $x(t_0)\in\mathbb{R}^n$ the solution of the system satisfies $\lim_{t\rightarrow t_0+\tau^-}\|x(t)\|\leq\varsigma$.
\item \textit{globally (uniformly) practically prescribed-time stable} if it is globally (uniformly) stable in the sense of Lyapunov and globally practically prescribed-time attractive. 
\end{enumerate}
\end{definition}

\begin{corollary}
\label{corollary:1}
Suppose that Assumption \ref{ass:1} holds. Let $t_0$ denote the initial time, $\Delta t=t-t_0$, $\mu\in\mathcal{M}(\tau)$, and $\kappa=\mu^{-1}\in\mathcal{K}(\tau)$. The solution of system \eqref{eq:1} under $u=\pi(x,t,\tau)$ for every $f(x,u,t)\in\mathcal{F}(x,u,t)$ and $\tau\in\mathcal{T}$ is
\begin{enumerate}
\item globally (uniformly) stable if and only if for every $f_{\mu}(\xi,t)\in\mathcal{F}(\xi_{\mu},u_{\mu},t_{\mu})$ and $\tau\in\mathcal{T}$ system \eqref{eq:th:1-4} is globally (uniformly) stable and for every $\xi(0)\in\mathbb{R}^n$ there exists $\sigma(\xi(0))>0$ such that the solution of system \eqref{eq:th:1-4} satisfies $\left\|\mathfrak{B}_n[\mu(\Delta t)]\xi(t)\right\|\leq\sigma(\xi(0))$ for all $t\in[t_0,\infty)$.
\item globally practically prescribed-time attractive at every $\tau\in\mathcal{T}$ with error $\varsigma>0$ if and only if for every $f_{\mu}(\xi,t)\in\mathcal{F}(\xi_{\mu},u_{\mu},t_{\mu})$ and $\tau\in\mathcal{T}$ the solution of system \eqref{eq:th:1-4} satisfies:
\begin{equation}
\label{eq:cor:1-1}
\lim_{t\rightarrow\infty}\left\|\mathfrak{B}_n[\mu(\Delta t)]\xi(t)\right\|\leq\varsigma.
\end{equation}
\item globally (uniformly) practically prescribed-time stable if and only if items 1 and 2 are met. 
\end{enumerate}
\end{corollary}
\begin{proof}
Item 1: Consider $h(x)=x$ and define the set $\mathcal{H}=\{h:\|h\|<\delta\}$. One can verify that according to Theorem \ref{th:1}, if for every $f_{\mu}(\xi,t)\in\mathcal{F}(\xi_{\mu},u_{\mu},t_{\mu})$ and $\tau\in\mathcal{T}$ system \eqref{eq:th:1-4} satisfies $h(\xi_\mu)=\mathfrak{B}_n[\mu(\Delta t)]\xi\in\mathcal{H}$ for $\xi(t_0)=\mathfrak{B}_n[\kappa(\Delta t)]_{t=t_0}x_0$, then for every $f(x,u,t)\in\mathcal{F}(x,u,t)$ and $\tau\in\mathcal{T}$ system \eqref{eq:1} under $u=\pi(x,t,\tau)$ satisfies $h(x)\in\mathcal{H}$ but for $x(t_0)=x_0$. Therefore, under the conditions described above, since $\mathfrak{B}_n[\kappa(\Delta t)]_{t=t_0}$ is positive-definite and $\mathfrak{B}_n[\mu(\Delta t)]\xi$ is bounded, if for every $\delta>0$ there exists $\varepsilon>0$ such that the solution of system \eqref{eq:th:1-4} satisfies $\|\xi(t_0)\|=\|\mathfrak{B}_n[\kappa(\Delta t)]_{t=t_0}x_0\|<\varepsilon\Rightarrow h(\xi_\mu)\in\mathcal{H}$ (which means that system \eqref{eq:th:1-4} is stable), then for every $\delta$ there exists $\epsilon>0$ such that the solution of system \eqref{eq:1} satisfies $\|x(t_0)\|=\|x_0\|<\epsilon\Rightarrow h(x)\in\mathcal{H}$ (which means that system \eqref{eq:1} is stable). The uniformity can be concluded by observing the fact that $\pi(x,t,\tau)$ only depends on $\Delta t$. Item 2: The practical prescribed-time attractivity can be proved if one considers $h(x)=x$ and $\mathcal{H}(t)=\{h:\ \lim_{t\rightarrow t_0+\tau^-}\|h\|\leq\varsigma\}$. To satisfy $h(x)\in\mathcal{H}(t)$ for system \eqref{eq:1} under $u=\pi(x,t,\tau)$, system \eqref{eq:th:1-4} should satisfy $h(\xi_\mu)=\mathfrak{B}_n[\mu(\Delta t)]\xi\in\mathcal{H}(t_\mu)=\{h:\ \lim_{t_{\mu}\rightarrow t_0+\tau^-}\|h\|\leq\varsigma\}=\{h:\ \lim_{\mu(\Delta t)\rightarrow \tau^-}\|h\|\leq\varsigma\}=\{h:\ \lim_{t\rightarrow \infty}\|h\|\leq\varsigma\}$, that is equivalent to \eqref{eq:cor:1-1}. Item 3: The proof of this item is a direct result of Definition \ref{def:PTS}. 
\end{proof}

\begin{remark}
\label{rem:1}
Many PTCs proposed in the literature can be derived from Theorem \ref{th:1}\footnote{MATLAB\textsuperscript{\tiny\textregistered} codes and Simulink\textsuperscript{\tiny\textregistered} models for some PTCs can be found in \href{https://github.com/a-shakouri/prescribed-time-control}{https://github.com/a-shakouri/prescribed-time-control}}. For instance, the PTC with linear decay that is studied in \cite{shakouri2021prescribed} can be obtained by considering an exponential mapping function and implementing the conditions of Corollary \ref{corollary:2} with $\zeta(\Delta t)=1-\Delta t/\tau$. Other prescribed-time stable controllers can be obtained by the use of Corollary \ref{corollary:1} with different mapping functions to handle various disturbances with different growth rates. 
\end{remark}

\begin{remark}
\label{rem:1p}
A practically prescribed-time attractive system is just prescribed-time attractive when the error is zero $\varsigma=0$ (see \cite[Definition 3]{shakouri2021prescribedeul}). In this case, the singularity at $t=t_0+\tau$ is an inevitable problem for many PTCs. However, since the error $\varsigma$ is also a user-defined parameter in most cases, the practical PTC does not compromise the advantages of an ideal PTC in the application, while it does not suffer from the singularity at the convergence moment.
\end{remark}

Theorem \ref{th:1} and its corollaries lead to a process for PTC design, which is summarized in the following algorithm:
\begin{alg}
\label{alg:1}
Let Assumption \ref{ass:1} is satisfied. Identify the dimension of the system $n$ and the set $\mathcal{F}$. Reformulate the control objectives of the system (such as state/input/output constraints) as defined in \eqref{eq:prob:1-1} in terms of function $h$ and set $\mathcal{H}$. Do the following steps:
\begin{enumerate}
    \item Select a class $\mathcal{M}$ function $\mu$.
    \item According to the system dimension $n$ and function $\mu$, obtain the functionals $\frak{B}_n[s]$ and $\frak{b}_n[s]$ as described in \eqref{eq:III-3} and \eqref{eq:III-5}. 
    \item Design an infinite-time controller $\pi_0(\xi,t)$ such that system \eqref{eq:th:1-4} satisfies the control objective \eqref{eq:th:1-5} for every $f_\mu(\xi,t)\in\mathcal{F}(\xi_\mu,u_\mu,t_\mu)$. If unsuccessful, go back to step 1 and select a different class $\mathcal{M}$ function or tune its growth rate.
    \item Obtain $\kappa=\mu^{-1}$ and construct the PTC as described in \eqref{eq:th:1-1}--\eqref{eq:th:1-3}. 
\end{enumerate}
\end{alg}

The following example demonstrates how Algorithm \ref{alg:1} can be implemented to a simple system.
\begin{example}
\label{ex:4}
Consider a scalar first-order system $\dot{x}=a(x,u,t)+bu+gu$. Function $a(x,u,t)$ is an unknown non-vanishing but bounded disturbance. Constant $b$ is unknown and satisfies $|b|<|g|$. In terms of system \eqref{eq:1}, we have $f(x,u,t)=a(x,u,t)+bu$, and $\mathcal{F}$ contains all functions $a(x,u,t)+bu$ such that the conditions on $a(x,u,t)$ and $b$ are met. The control objectives are $|x(t)|\leq\phi|x(0)|$, $\phi>1$, and $\lim_{t\rightarrow\tau}|x(t)|=0$, which can be formulated as $h=|x(t)|$ and $\mathcal{H}=\{h:\ h\leq\phi|x(0)|,\ \lim_{t\rightarrow\tau}h\leq\sigma\}$. As the system is scalar, we have $n=1$, $\frak{B}_1=1$, and $\frak{b}_1=0$. We choose $\kappa=-\ln\left(1-t/\tau\right)$, and correspondingly $\mu=\tau(1-\exp(-t))$. Hence, system \eqref{eq:th:1-4} reduces to $\dot{\xi}=\tau \exp(-t)a+(b/g+1)\pi_0$ (which has a vanishing uncertainty and an always positive control gain). We choose an infinite-time controller as $\pi_0(\xi)=-\frac{\psi\xi}{(|\xi|-\phi|\xi(0)|)^2}$, $\psi>0$ that satisfies $|\xi(t)|\leq\phi|\xi(0)|$ and $\lim_{t\rightarrow\infty}\xi(t)=0$. Therefore, a PTC as $\pi(x,t)=-\frac{\psi x}{g(\tau-t)(|x|-\phi|x(0)|)^2}$ satisfies the control objectives of the main system. Fig. \ref{fig:1} shows the simulation results for $a(x,u,t)=0.1-t^3\exp(-t)\sin(x/(u+0.001))$, $g=1$, $b=-0.5$, $\phi=1.1$, $\psi=20$, and different values of convergence time $\tau$ corresponding to initial condition $x(0)=1$. 
\end{example}

\begin{figure}[!h]
\centering\includegraphics[width=1\linewidth]{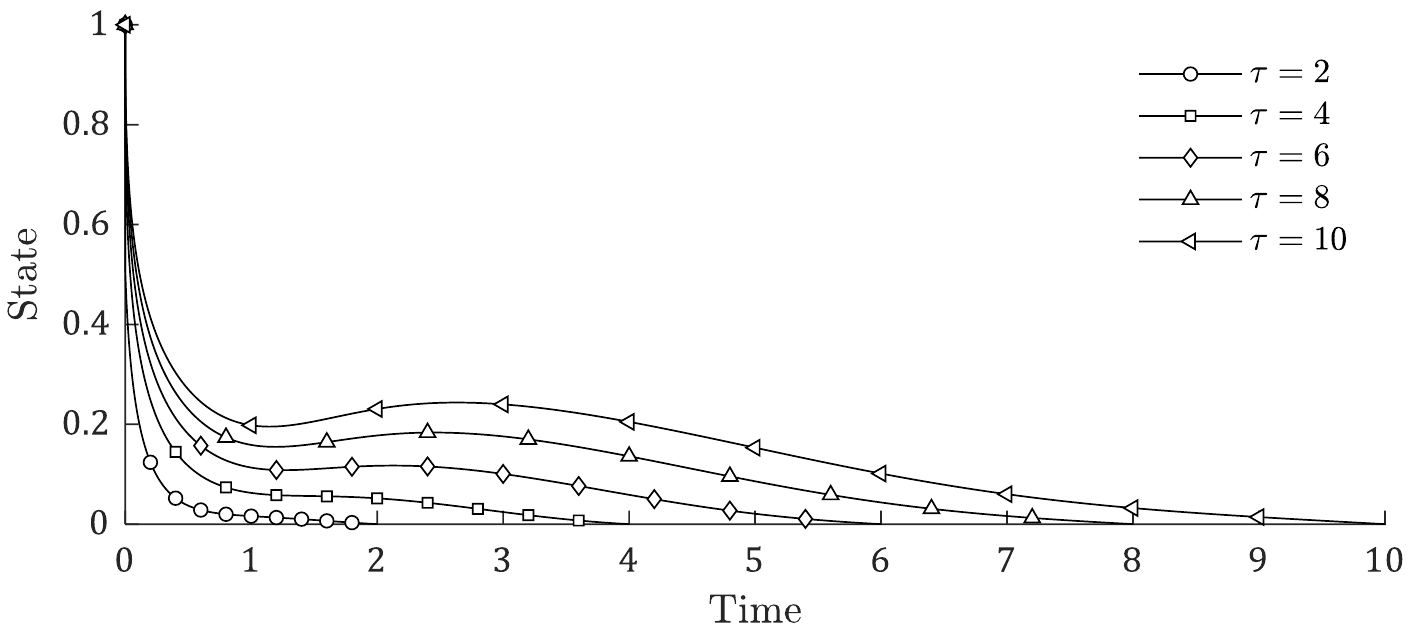}
\caption{Simulation results for Example \ref{ex:4}.}
\label{fig:1}
\end{figure}

\section{Proof of Theorem \ref{th:1}}
\label{sec:IV}

Without loss of generality, we prove Theorem \ref{th:1} for the case of unity input gain $g(x,t)=1$. In the first step, we deal with the nominal system that is a chain of integrators. Next, we find how the unmodeled term, $f(x,u,t)$, affects the obtained results. 

The following lemmas can be evaluated after considering Definitions \ref{def:1} and \ref{def:2}:

\begin{lemma}[see \cite{shakouri2021prescribedeul}]
\label{lemma:1}
The following statements hold for any functions $\kappa\in\mathcal{K}(\tau)$ and $\mu=\kappa^{-1}\in\mathcal{M}(\tau)$:
\begin{enumerate}
\item $\dot{\kappa}(\cdot):[0,\tau)\rightarrow[0,\infty)$, $\ddot{\kappa}(\cdot):[0,\tau)\rightarrow[0,\infty)$, and $\lim_{t\rightarrow\tau^-}\dot{\kappa}(t)=\lim_{t\rightarrow\tau^-}\ddot{\kappa}(t)=\infty$.
\item $\dot{\mu}(\cdot):[0,\infty)\rightarrow[0,\infty)$, $\ddot{\mu}(\cdot):[0,\infty)\rightarrow(-\infty,0]$, and $\lim_{t\rightarrow\infty}\dot{\mu}(t)=-\lim_{t\rightarrow\infty}\ddot{\mu}(t)=0$. 
\item The sum of two class $\mathcal{K}$ functions belongs to class $\mathcal{K}$, and the sum of two class $\mathcal{M}$ functions divided by $2$ belongs to class $\mathcal{M}$.
\item  Function $\dot{\mu}^{\alpha}(t)$ for $\alpha\geq1$ is the derivative of a class $\mathcal{M}$ function.
\item Any exponential function of the form $\beta \exp(-\alpha t)$ with arbitrary $\alpha>0$ can be the derivative of a class $\mathcal{M}(\tau)$ function if and only if $\beta=\alpha\tau$.
\end{enumerate}
\end{lemma}

Henceforth, for a function $s(\cdot)$, we use the following derivative notations: $\dot{s}=\mathrm{d}s/\mathrm{d}t$ and $\acute{s}=\mathrm{d}s/\mathrm{d}\mu$. To take the role of initial time $t_0$ into account, let us define the operator $\Delta$ that acts on any variable $s$ as $\Delta s=s-t_0$ (similar to $\Delta t=t-t_0$). Also, let us define functions $\tilde{\mu}(\cdot):[t_0,\infty)\rightarrow[t_0,t_0+\tau)$ and $\tilde{\kappa}(\cdot):[t_0,t_0+\tau)\rightarrow[t_0,\infty)$ that are functions $\mu\in\mathcal{M}$ and $\kappa\in\mathcal{K}$ which are shifted $t_0$ units up and right:
\begin{subequations}
\begin{align}
\label{eq:IV-1}
\tilde{\mu}(t)&=\mu(\Delta t)+t_0,\\
\label{eq:IV-2}
\tilde{\kappa}(\tilde{\mu})&=\kappa(\Delta \tilde{\mu})+t_0.
\end{align}
\end{subequations}
One can verify that $\tilde{\mu}$ is the inverse function of $\tilde{\kappa}$, $\Delta\tilde{\mu}=\mu(\Delta t)$, and $\Delta\tilde{\kappa}=\kappa(\Delta \tilde{\mu})$. Since $\dot{\tilde{\mu}}(t)=\dot{\mu}(\Delta t)$ and $\acute{\tilde{\kappa}}(\tilde{\mu})=\acute{\kappa}(\Delta\tilde{\mu})$, we use the derivatives of $\mu(\Delta t)$ and $\kappa(\Delta\tilde{\mu})$ instead of the derivatives of $\tilde{\mu}(t)$ and $\tilde{\kappa}(\tilde{\mu})$, respectively. 

Consider a chain of integrators under an input $u=\pi_0(\xi,t)$ as follows:
\begin{equation}
\label{eq:IV-3}
\left\{
\begin{array}{lcl}
\dot{\xi}_i&=&\xi_{i+1};\quad i=1,2,\ldots,n-1 \\
\dot{\xi}_n&=&\pi_0(\xi,t)
\end{array}\right..
\end{equation}

Let us map the solutions of system \eqref{eq:IV-3} from $[t_0,\infty)$ onto $[t_0,t_0+\tau)$ with the following equivalent rules:
\begin{subequations}
\begin{align}
\label{eq:IV-4}
\chi_1(\tilde{\mu}(t))&=\xi_1(t),\hspace{4mm}&t\in[t_0,\infty),\\
\label{eq:IV-5}
\xi_1(\tilde{\kappa}(\tilde{\mu}))&=\chi_1(\tilde{\mu}),\hspace{4mm}&\tilde{\mu}\in[t_0,t_0+\tau).
\end{align}
\end{subequations}
Furthermore, we define $\mathrm{d}\chi_i/\mathrm{d}\tilde{\mu}=\chi_{i+1}$ and $\chi=[\chi_1,\ldots,\chi_n]^T$. For evaluating the derivatives of higher orders, consider the following lemmas:

\begin{lemma}[Fa\`a di Bruno's formula]
\label{lemma:2}
Suppose $s_1(\cdot)$ and $s_2(\cdot)$ are sufficiently differentiable functions. Then: 
\begin{equation}
\label{eq:lem:2-1}
\begin{split}
\frac{\mathrm{d}^i}{\mathrm{d}t^i}s_1(s_2(t))=\sum_{j=1}^is_1^{(j)}(s_2(t))B_{i,j}(s_2^{(1)}(t),\ldots,s_2^{(i-j+1)}(t)).
\end{split}
\end{equation}
\end{lemma}
\begin{proof}
This lemma is the Bell polynomial version of the well-known Fa\`a di Bruno's formula, which is also known as Riordan's formula \cite{fa}. 
\end{proof}

\begin{lemma}
\label{lemma:inv}
Let $s_1(\cdot)$ and $s_2(\cdot)$ are sufficiently differentiable and invertible functions such that $s_1$ is the inverse function of $s_2$, i.e., $s_1(s_2(t))=t$. The following formula relates the $k$th-order derivative of $s_2(t)$ to the derivatives of $s_1(s_2)$: 
\begin{equation}
\label{eq:lem:inv-1}
s_2^{(k)}(t)=r_k(s_1(s_2)),
\end{equation}
where functional $r_k$ is defined in \eqref{eq:III-4}.
\end{lemma}
\begin{proof}
Equation \eqref{eq:lem:inv-1} is an extension of the Fa\`a di Bruno's formula presented in Lemma \ref{lemma:2} for inverse functions. This version is proved in \cite{combinv}.  
\end{proof}

\begin{lemma}
\label{lemma:3}
Let $\mathrm{d}\chi_i/\mathrm{d}\tilde{\mu}=\chi_{i+1}$, $i=1,\ldots,n$. Suppose $\tilde{\mu}$ and $\chi_1$ are sufficiently differentiable functions and equations \eqref{eq:IV-3} and \eqref{eq:IV-4} hold. Then: 
\begin{subequations}
\begin{align}
\label{eq:lem:3-1a}
\xi(t)&=\mathfrak{B}_n[\kappa(\Delta \tilde{\mu})]\chi(\tilde{\mu}),\\\label{eq:lem:3-1b}
\chi(\tilde{\mu})&=\mathfrak{B}_n[\mu(\Delta t)]\xi(t).
\end{align}
\end{subequations}
\end{lemma}

\begin{proof}
To prove \eqref{eq:lem:3-1a}, observe that given Lemma \ref{lemma:2}, by setting $s_1=\chi_i$ and $s_2=\tilde{\mu}(t)$, one obtains: 
\begin{equation}
\label{eq:lem:3-3}
\begin{split}
\xi_{i+1}(t)&=\chi_{i+1}(\tilde{\mu}(t))=\frac{\mathrm{d}^{i}}{\mathrm{d}t^{i}}\chi(\tilde{\mu}(t))\\
&=\sum_{j=1}^{i}\chi^{(j)}(\tilde{\mu})B_{i,j}(\tilde{\mu}^{(1)}(t),\ldots,\tilde{\mu}^{(i-j+1)}(t)).
\end{split}
\end{equation}
From Lemma \ref{lemma:inv}, we have $\tilde{\mu}^{(k)}(t)=r_k(\kappa(\mu))$. Therefore, using definition \eqref{eq:III-3} and knowing that $\chi^{(j)}(\tilde{\mu})=\chi_{j+1}$, a substitution yields identity \eqref{eq:lem:3-1a}. The same procedure can be utilized for the proof of \eqref{eq:lem:3-1b}. 
\end{proof}

Observe that formulas \eqref{eq:lem:3-1a} and \eqref{eq:lem:3-1b} can also be stated in the following component forms:
\begin{subequations}
\begin{align}
\label{eq:IV-8a}
\dot{\xi}_{n}(t)&=\frac{1}{\acute{\kappa}^{n}(\Delta \tilde{\mu})}\acute{\chi}_{n}(\tilde{\mu})+\mathfrak{b}_n^T[\kappa(\Delta \tilde{\mu})]\chi(\tilde{\mu}),\\\label{eq:IV-8b}
\acute{\chi}_n(\tilde{\mu})&=\frac{1}{\dot{\mu}^n(\Delta t)}\dot{\xi}_n(t)+\mathfrak{b}_n^T[\mu(\Delta t)]\xi(t),
\end{align}
\end{subequations}
which yields the following identity:
\begin{equation}
\label{eq:IV-8p}
\dot{\mu}^n(\Delta t)\mathfrak{b}_n^T[\mu(\Delta t)]\xi(t)=-\mathfrak{b}_n^T[\kappa(\Delta \tilde{\mu})]\chi(\tilde{\mu}).
\end{equation}
Moreover, one can verify that the infinite time scale $t$ can be expressed in terms of a class $\mathcal{K}$ function with argument $\tilde{\mu}$ as follows:
\begin{equation}
\label{eq:IV-9}
t=\kappa(\Delta \tilde{\mu})+t_0.
\end{equation}
Substituting \eqref{eq:lem:3-1a}, \eqref{eq:IV-8a}, and \eqref{eq:IV-9} in system \eqref{eq:IV-3}, the following system can be obtained in terms of $\chi$ and $\tilde{\mu}$:
\begin{equation}
\label{eq:IV-10}
\left\{
\begin{array}{lcl}
\acute{\chi}_i&=&\chi_{i+1};\hspace{4mm}i=1,2,\ldots,n-1 \\
\acute{\chi}_n&=&\acute{\kappa}^n(\Delta \tilde{\mu})\left(\pi_0(\chi_{\kappa},\tilde{\mu}_{\kappa})-\mathfrak{b}_{n}^T[\kappa(\Delta \tilde{\mu})]\chi\right)
\end{array}\right.,
\end{equation}
where $\chi_{\kappa}=\mathfrak{B}_n[\kappa(\Delta \tilde{\mu})]\chi$ and $\tilde{\mu}_{\kappa}=\kappa(\Delta \tilde{\mu})+t_0$. System \eqref{eq:IV-10} imitates the infinite-time behavior of system \eqref{eq:IV-3} in a prescribed finite time interval of $\tau$ such that the maximum and minimum of $\xi_1$ remains the same for $\chi_1$. Observe that system \eqref{eq:IV-10} is the closed-loop response of system \eqref{eq:1} with $f=0$ and $g(x,t)=1$ under $u=\pi(x,t,\tau)$, when $x$ and $t$ are replaced by $\chi$ and $\tilde{\mu}$.

Applying the disturbance term $f$ to system \eqref{eq:IV-10} and inverse mapping it from $[t_0,t_0+\tau)$ onto $[t_0,\infty)$, yields the following system:
\begin{equation}
\label{eq:IV-11}
\left\{
\begin{array}{lcl}
\dot{\xi}_i&=&\xi_{i+1};\quad i=1,2,\ldots,n-1 \\
\dot{\xi}_n&=&\dot{\mu}^n(\Delta t)f(\xi_{\mu},u_{\mu},t_{\mu})+\pi_0(\xi,t)
\end{array}\right.,
\end{equation}
where $\xi_{\mu}$, $u_{\mu}$, and $t_{\mu}$ are described in \eqref{eq:th:1-6}, \eqref{eq:th:1-7}, and \eqref{eq:th:1-8}, respectively. Therefore, system \eqref{eq:IV-11} is a chain of integrators under a (most likely vanishing) matched disturbance $\dot{\mu}^n(\Delta t)f(\xi_{\mu},u_{\mu},t_{\mu})$ with $\pi_0(\xi,t)$ as an input signal. Given Assumption \ref{ass:1}, we know that $f(\xi_{\mu},u_{\mu},t_{\mu})\in\mathcal{F}(\xi_{\mu},u_{\mu},t_{\mu})$. Since variables $\xi_{\mu}$, $u_{\mu}$, $\pi_0(\xi,t)$, and $t_{\mu}=\tilde{\mu}(t)$ are functions of $\xi$ and $t$, we rename  $f(\xi_{\mu},u_{\mu},t_{\mu})$ as $f_{\mu}(\xi,t)$ to be directly expressed in terms of the system variables. As the final step, verify that according to the developed formulas, an output signal of system \eqref{eq:IV-11} written as $h(\xi_\mu,u_\mu,t_\mu)$ evaluated at $t_\mu$, is equal to output signal $h(x,u,t)$ of system \eqref{eq:1} under $u=\pi(x,t,\tau)$ evaluated at $t$.

\section{Conclusions}    
\label{sec:VI}

A general time-transformation methodology has been formulated, which unifies the mapping approaches into a universal representation towards prescribed-time control (PTC) design. In the light of the presented framework, it has been shown that a nonasymptotic PTC design subject to any nonlinear control objectives can be viewed as a traditional asymptotic control design for the associated infinite-time system under a new set of control objectives.




\bibliographystyle{IEEEtran}
\bibliography{biblo}

\end{document}